\documentclass[reqno]{amsart}

\usepackage{booktabs}
\usepackage{IEEEtrantools}
\usepackage{amssymb,latexsym,amsfonts,amsmath}
\usepackage{mathrsfs}
\usepackage{graphicx}
\usepackage{dsfont}
\usepackage{tikz}
\usetikzlibrary{calc,shapes,arrows}

\usepackage{algorithm}
\usepackage[noend]{algpseudocode}

\usepackage{xspace}

\newtheorem{theorem}{Theorem}[section]

\newtheorem{problem}[theorem]{Problem}

\newtheorem{definition}[theorem]{Definition}

\newtheorem{remark}[theorem]{Remark}
\newtheorem{assumption}{Assumption}
\numberwithin{equation}{section}

\newcommand{\R}{{\mathbb{R}}}

\newcommand{\EE}{\mathds{E}}
\newcommand{\PP}{\mathds{P}}

\usepackage{fancyhdr}

\newenvironment{nouppercase}{%
	\renewcommand{\uppercasenonmath}[1]{}}{}

\usepackage{amsmath}
\usepackage[many]{tcolorbox}
\usetikzlibrary{calc}
\tcbuselibrary{skins}

\newtcolorbox{resp}[1][]{%
	enhanced jigsaw,%
	colback=gray!5!white,%
	colframe=gray!80!black,%
	size=small,%
	boxrule=1pt,%
	halign title=flush center,%
	coltitle=black,%
	breakable,%
	drop shadow=black!50!white,%
	attach boxed title to top left={xshift=1cm,yshift=-\tcboxedtitleheight/2,yshifttext=-\tcboxedtitleheight/2},%
	minipage boxed title=3cm,%
	boxed title style={%
		colback=white,%
		size=fbox,%
		boxrule=1pt,%
		boxsep=2pt,%
		underlay={%
			\coordinate (dotA) at ($(interior.west) + (-0.5pt,0)$);
			\coordinate (dotB) at ($(interior.east) + (0.5pt,0)$);
			\begin{scope}[gray!80!black]
				\fill (dotA) circle (2pt);
				\fill (dotB) circle (2pt);
			\end{scope}
		}%
	},%
	#1%
}

\begin{document}

\begin{abstract}
This work is concerned with a formal approach for safety controller synthesis of stochastic control systems with both process and measurement noises while considering \emph{wireless communication networks} between sensors, controllers, and actuators. The proposed scheme is based on \emph{control barrier certificates (CBC)}, which allows us to provide safety certifications for wirelessly-connected stochastic control systems. Despite the available literature on designing control barrier certificates, there has been unfortunately no consideration of wireless communication networks to capture potential packet losses and end-to-end delays, which is absolutely crucial in safety-critical real-world applications. In our proposed setting, the key objective is to construct a control barrier certificate together with a safety controller while providing a lower bound on the satisfaction probability of the safety property over a finite time horizon. We propose a systematic approach in the form of sum-of-squares optimization and matrix inequalities for the synthesis of CBC and its associated controller. We demonstrate the efficacy of our approach on a permanent magnet synchronous motor. For the application of automotive electric steering under a wireless communication network, we design a CBC together with a safety controller to maintain the electrical current of the motor in a safe set within a finite time horizon while providing a formal probabilistic guarantee.
\end{abstract}

\title{{\LARGE Safety Barrier Certificates for Stochastic Control Systems with Wireless Communication Networks}}

\author{{\bf {\large Omid Akbarzadeh}}}
\author{{\bf {\large Sadegh Soudjani}}}
\author{{{\bf {\large Abolfazl Lavaei}}}\\
	{\normalfont School of computing, Newcastle University, United Kingdom}\\
\texttt{\{o.akbarzadeh2, sadegh.soudjani, abolfazl.lavaei\}@newcastle.ac.uk}}

\pagestyle{fancy}
\lhead{}
\rhead{}
  \fancyhead[OL]{Omid Akbarzadeh, Sadegh Soudjani, Abolfazl Lavaei}

  \fancyhead[EL]{Safety Barrier Certificates for Stochastic Control Systems with Wireless Communication Networks}
  \rhead{\thepage}
 \cfoot{}
 
\begin{nouppercase}
	\maketitle
\end{nouppercase}

\section{Introduction}
This paper is mainly motivated by the difficulties arising in the safety controller design of complex stochastic systems with wireless communication networks while capturing potential packet losses and end-to-end delays. The complexity of controller synthesis over  stochastic systems is not only due to continuous state and input sets with potentially high dimensions, but also due to the stochastic nature of dynamics that is capable of capturing different sources of uncertainty in both model and environment. Consequently, synthesizing a formal controller for such complicated systems to fulfill high-level logic specifications, \emph{e.g.,} linear temporal logic (LTL) properties~\cite{baier2008principles}, is naturally a critical problem. To alleviate the aforesaid complexities, existing results in the literature have been mainly reliant on employing finite abstractions to approximate original (a.k.a. concrete) models with simpler ones with discrete state sets (see \emph{e.g.,}~\cite{ref1,julius2009approximations,zamani2015symbolic}). However, the proposed abstraction-based approaches are based on discretizing state and input sets and may not be applicable to many real-world scenarios due to the state-explosion problem. Hence, compositional techniques have been proposed over the past recent years to build finite abstractions of high-dimensional stochastic systems via those of smaller-scale subsystems (see \emph{e.g.,}~\cite{SAM17,lavaei2020compositional,hahn2013compositional,lavaei2022dissipativity,ref3,ref4,nejati2021compositional,lavaei2022scalable}).

Another promising solution for formal verification and control synthesis of stochastic systems is to use control barrier certificate, as a \emph{discretization-free approach}, initially proposed in~\cite{ref5,prajna2004safety}. Generally speaking, barrier certificates are Lyapunov-like functions whose main goal is to fulfill some conditions on the barrier function itself and its evaluation over one-step transition of the system. Given a set of initial states, an initial level set of barrier certificates separates an unsafe region from trajectories of the system. Consequently, the presence of such barrier certificates provides a formal probabilistic guarantee over the safety of the system (cf. Fig.~\ref{fig.12}). Barrier certificates have been extensively used for formal analysis of complex stochastic systems (see \emph{e.g.,}~\cite{ref7, ref8, ref9, ref10, ref11, ref12, ref13,ref6,anand2022small,lavaei2022compositional,nejati2022compositional,nejati2022dissipativity}).
\begin{figure}[t!]
	\includegraphics[scale=0.55]{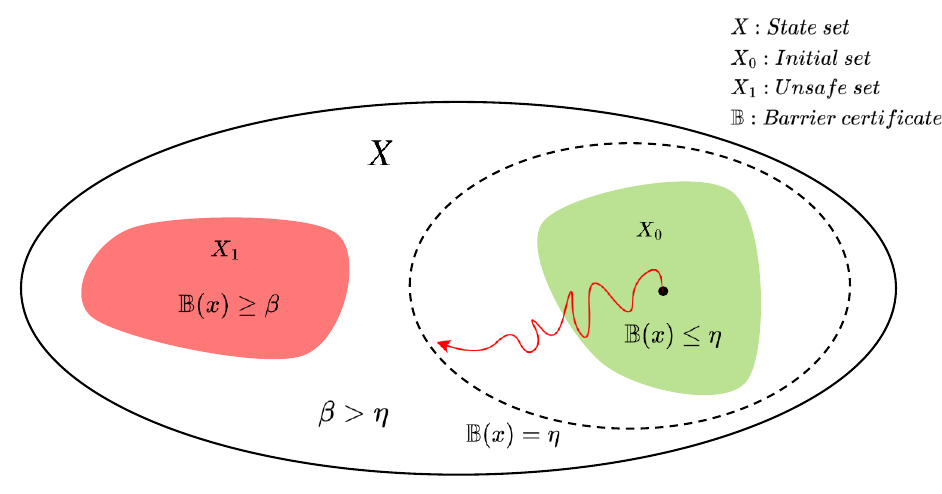}
	\centering
	\caption{A control barrier certificate for stochastic control systems. The dashed line represents the initial level set $\mathbb{B}(x) = \eta$.}
	\label{fig.12}
\end{figure}

Unfortunately, although the proposed approaches on both abstraction-based and discretization-free (\emph{i.e.,} control barrier certificate) techniques are very comprehensive in the relevant literature, none of those approaches consider wireless communications to capture potential packet losses and end-to-end delays, while enforcing high-level logic properties (\emph{e.g.,} LTL specifications). In particular, establishing an effective communication in control systems is crucial as it ensures that information is transmitted accurately and timely between various components of underlying systems (\emph{e.g.,} sensor-to-controller and controller-to-actuator branches of the network). This can facilitate efficient decision-making, troubleshooting, and overall performance of control systems. Currently, it is a widespread practice to establish connections between physical processes and remote control stations via \emph{wired communication}. Nevertheless, proposed approaches based on cable impose restrictions on flexibility and elevate the expenses of both installation and maintenance for the entire system. In particular, compared to \emph{wireless communications}, cables are prone to experiencing wear and tear, ultimately leading to their failure after a certain period of time. This phenomenon mainly results in inaccuracies that are difficult to detect, and accordingly, diminish efficiency.

The key contribution of this paper is to propose a formal approach for safety controller synthesis of stochastic linear control systems with \emph{wireless communications} between sensor-to-controller and controller-to-actuator branches using a \emph{Glossy} network architecture~\cite{ref24}. In our proposed framework, we construct a control barrier certificate together with a safety controller while providing a lower bound on the satisfaction probability of the safety property within a finite time horizon. We also provide a systematic approach in the form of sum-of-squares optimization and matrix inequalities for the construction of CBC and its corresponding controller. We verify the fruitfulness of our approach on a permanent magnet synchronous motor. Given a wireless communication network in the system, we design a CBC together with a safety controller to maintain the electrical current of the motor in an interval within a bounded time horizon while providing a formal probabilistic guarantee.

There has been a limited number of work on formal analysis of cyber-physical systems with wireless communications. Existing results include formal verification and design of wireless cyber-physical systems~\cite{ 7, ref16,ref22}, and \emph{stability analysis} of control systems with wireless communication networks~\cite{ref17,ref19,ref20}. It is worth mentioning that the type of wireless communication networks, employed in our work, is similar to the one used in~\cite{ref17}. However, our results are developed for \emph{stochastic systems} with both process and measurement noises while ensuring \emph{safety specifications} instead of considering stability properties over deterministic systems, which is the case in~\cite{ref17}. To the best of our knowledge, this work is the first to propose a notion of control barrier certificates for ensuring \emph{safety certifications} over \emph{wirelessly connected stochastic control systems}.

The remainder of this article is formalized as follows. In Section~\ref{sec:discrete}, we present required mathematical notations and a formal definition of discrete-time stochastic linear control systems. In Section~\ref{sec:wireless}, we present wireless communication networks together with our controller synthesis framework. In Section~\ref{sec:CBC}, we introduce control barrier certificates for stochastic control systems with wireless communication networks and propose a systematic approach to design CBC and its associated controller while ensuring the desired safety requirements. Finally, in Section~\ref{sec:case}, the results of the paper are illustrated over a permanent magnet synchronous motor, followed by a concluding discussion and future work direction in Section~\ref{conclude}.

\section{Discrete-Time Stochastic Control Systems}{\label{sec:discrete}}

\subsection{Notation and Preliminaries}
We present the sets of real, positive, and non-negative real numbers using, respectively, $\mathbb{R}$, $\mathbb{R}_{>0}$, and $\mathbb{R}_{\geq 0}$. Symbols $\mathbb{N}:=\{0,1,2,\dots\}$ and $\mathbb{N}_{\geq 1}:=\{1,2,\dots\}$ denote, respectively, sets of non-negative and positive integers. A column vector, given $N$ vectors $x_i \in \mathbb{R}^{n_i}$, is represented by $x=[x_1;\dots;x_N]$. We denote $\text{Tr}(A)= \sum_{i=1}^N a_i$, given a matrix $A \in \mathbb{R}^{N \times N}$ with diagonal elements $a_1, \dots , a_N $. An identity matrix in $\R^{n\times{n}}$ is denoted by $\mathbf{I}_n$. A normal distribution is denoted by $\mathcal{N}(\mu,\Sigma)$, with $\mu$ being a mean value and $\Sigma$ its covariance matrix. Symbol $\mathbf{0}_{n\times m}$ denotes a zero matrix in $\mathbb R^{n\times{m}}$, \emph{i.e.,} all elements are equal to zero. 

We denote the probability space of events by $(\Omega, \mathcal{F}_{\Omega}, \PP_{\Omega})$, with $\Omega$ being a sample space, $\mathcal{F}_{\Omega}$ a sigma-algebra on $\Omega$, and $\PP_{\Omega}$ a probability measure. We assume random variables $\mathcal X$ are measurable functions, \emph{i.e.,} $\mathcal X:(\Omega,\mathcal F_{\Omega})\rightarrow (S_{\mathcal X},\mathcal F_{\mathcal X})$, in a way that $\mathcal X$ induces a probability measure on $(S_{\mathcal X},\mathcal F_{\mathcal X})$ as $Prob\{\mathcal A\} = \PP_{\Omega}\{\mathcal X^{-1}(\mathcal A)\}$ for all $\mathcal A\in \mathcal F_{\mathcal X}$.
A topological space  $\mathbf{S}$ is Borel, denoted by $\mathcal B(\mathbf{S})$, if it is homeomorphic to a Borel subset of a Polish space, \textit{i.e.}, a separable and metrizable space.

\subsection{Discrete-Time Stochastic Control Systems}
In this work, we focus on discrete-time stochastic linear control systems as presented in the next definition.

\begin{definition}\label{dt-SLS}
  A discrete-time stochastic linear control system (dt-SLS) is defined by the tuple
   \begin{equation}
    \mathbb{S}=(X,U,Y,A,B,w_1,w_2),
    \label{eq.1}
\end{equation} 
 
\noindent where:
\begin{itemize}
    \item $X \subseteq \mathbb{R}^n$ is a Borel space as the state set; 
    \item $U \subseteq \mathbb{R}^m$ is a Borel space as the input set;
    \item $Y \subseteq \mathbb{R}^n$ is a Borel space as the output set;
     \item $A\in \mathbb R^{n\times n}$ and $B\in \mathbb R^{n\times m}$ are, respectively, state and input matrices of dt-SLS;
    \item $w_1$ denotes a sequence of independent-and-identically distributed (i.i.d.) random variables as the systems' process noise;
    \item $w_2$ denotes a sequence of i.i.d. random variables as the sensors' measurement noise.
\end{itemize}
\end{definition}

\noindent Given an initial state $x_0 = x(0) \in X$ and under an input signal $u(\cdot): \mathbb N \to U$, the evolution of the state of dt-SLS $\mathbb{S}$ can be described as

\begin{equation}
\mathbb{S} :
    \begin{cases}
        x(k+1) = A x(k) + B u(k) + w_1(k), \\
      y(k) = x(k)+ w_{2}(k),\quad\quad k\in \mathbb N.
    \end{cases}  
    \label{eq.12}
\end{equation}

\noindent The random sequence $x_{x_0 u}\!: \Omega\times \mathbb N \to X$ satisfying (\ref{eq.12}) for any initial state $x_0 \in X$ , under an input signal $u(\cdot): \Omega \to U$, is called \emph{solution process} of dt-SLS at time $k \in \mathbb{N}$.

Next, we formally define the safety problem for dt-SLS $\mathbb{S}$ in~\eqref{eq.12} that we aim ultimately to solve.
\begin{definition}\label{Safety}
	Given a safety specification $\Upsilon = (X_0, X_1, \mathcal T)$, with $X_0,X_1\subseteq X$ being, respectively, initial and unsafe sets of dt-SLS ($X_0\cap X_1 = \emptyset$), the dt-SLS $\mathbb{S}$ is said to be safe within time horizon $\mathcal{T} \in \mathbb{N}\cup\{\infty\}$, denoted by $\mathbb{S}\models\Upsilon$, if trajectories of $\mathbb{S}$ starting from $X_{0}\subseteq X$ do not reach $X_1\subseteq X$ during  $\mathcal{T}$\!. Since trajectories
	of dt-SLS $\mathbb{S}$ are probabilistic, we are interested in
	computing $\PP \{\mathbb{S}\models\Upsilon\}\ge 1-\varepsilon$, where $\varepsilon \in[0,1]$.
\end{definition}

The subsequent section is dedicated to the wireless network architecture together with its controller synthesis framework. In particular, we aim to formulate communications between different components of dt-SLS including sensors, controller, and actuators over a wireless network. 

\begin{figure}[h!]
	\includegraphics[scale=0.7]{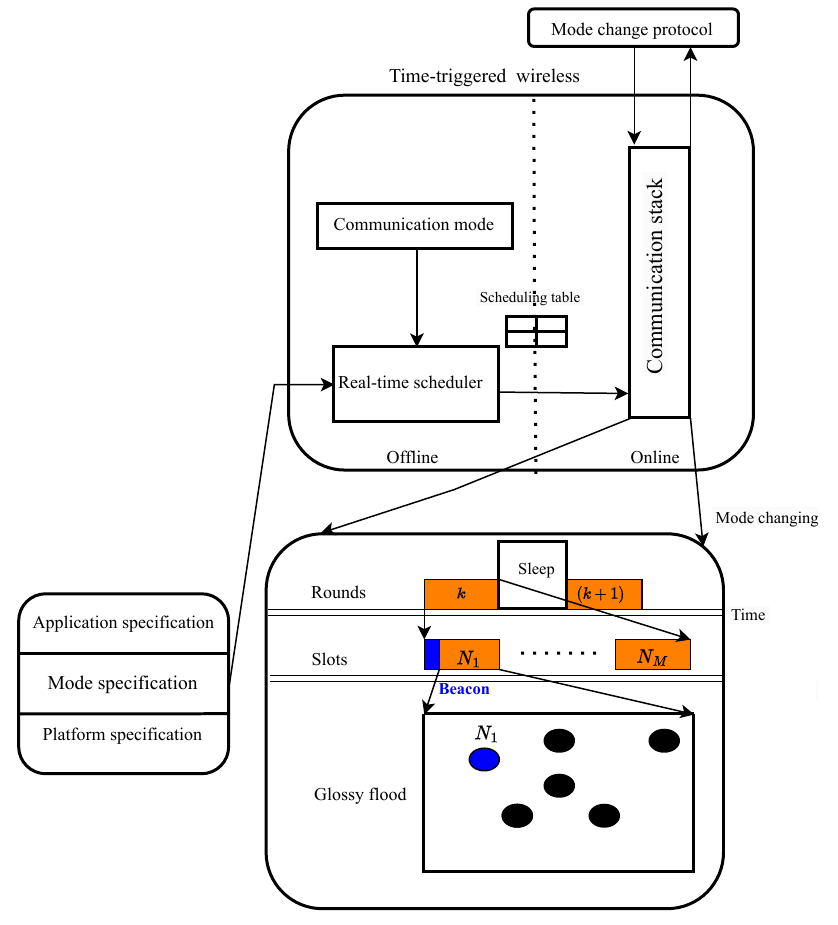}
	\centering
	\caption{A wireless communication protocol based on Glossy. This protocol operates as various periodic communication rounds, in which each round includes a sequence of non-overlapping time slots. All nodes of the network should get involved in a Glossy flood in every time slot, where a message is transmitted from a node to all other nodes.}
	\label{fig.2}\vspace{0.3cm}
\end{figure}

\section{Wireless Communication Network}{\label{sec:wireless}}
In this paper, we employ a multi-hop wireless protocol that is able to support many-to-all communications while satisfying requirements for a potentially low packet loss and end-to-end delay \cite{ref23}. It is shown that a network architecture based on \emph{Glossy}
can fulfill potential requirements with a high efficiency \cite{ref24}. Glossy is a flooding architecture for wireless sensor networks in which its performance is independent of node density, which facilitates its application in a variety of real-world scenarios. In a multi-hop wireless network, a node refers to a device or sensor that is part of the network and is capable of transmitting, receiving, and relaying data packets over the network. Furthermore, on top of Glossy, a suitable protocol is designed using a time-triggered wireless as a scheduling framework that can reduce an end-to-end delay to a few tens of milliseconds for many-to-all communications \cite{ref17, ref25}. A high-level overview of this protocol is depicted in Fig.~\ref{fig.2}. As it can be observed, this protocol operates as different periodic communication rounds, in which each round includes a sequence of non-overlapping time slots. In each time slot, all nodes of the network should get involved in a Glossy flood, where a message is transmitted from a node to all other nodes. A beacon slot is initiated by a dedicated node to synchronize at the beginning of each communication round. Under this wireless communication network, we formalize the main problem that we want to solve. \vspace{0.2cm}

\begin{resp}
	\begin{problem}\label{problem1}
		Consider a dt-SLS $\mathbb{S}$ with a Glossy wireless communication protocol and a safety specification $\Upsilon = (X_0, X_1, \mathcal T)$. Synthesize a formal controller based on control barrier certificates to ensure the satisfaction of safety specification $\Upsilon$ within a time horizon $\mathcal T$ with a guaranteed probabilistic bound $\varepsilon \in[0,1]$, \emph{i.e.,}\vspace{-0.2cm}
		\begin{align*}
		\PP \Big\{\mathbb{S}\models\Upsilon\Big\}\ge 1-\varepsilon.
		\end{align*}
	\end{problem}
\end{resp}
\begin{remark}
    The proposed multi-hop wireless network assumes that the losses for \emph{measurement and control} packets occur independently in a statistically uncorrelated manner. Specifically, it is assumed that the occurrence of a loss for a measurement packet does not affect the probability of a loss for a control packet, and vice versa.
\end{remark}
   
To address Problem~\ref{problem1}, inspired by~\cite{ref17}, we present our controller synthesis framework in the next subsection.
\begin{figure}[t]
	\includegraphics[scale=0.6]{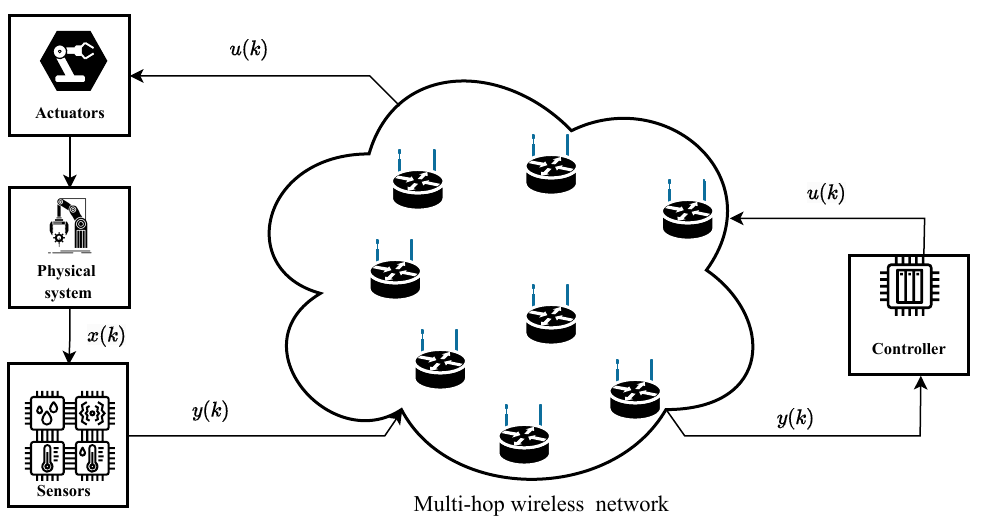}\vspace{0.1cm}
	\centering
	\caption{An overview of stochastic control systems with different components under a multi-hop wireless communication network.}
	\label{fig.1}\vspace{0.3cm}
\end{figure}

\subsection{Controller Synthesis Framework}
Here, we present our controller synthesis framework under a multi-hop wireless communication network, as depicted in Fig~(\ref{fig.1}). The underlying framework consists of sensors for monitoring systems' states, wireless transmission of data to controllers, processing of data using control algorithms, and wireless transmission of control signals to physical processes. The main goal of the controller synthesis is to maintain all trajectories of dt-SLS safe in the sense of Definition~\ref{Safety} in the presence of both process and measurement noises. In our controller framework depicted in Fig.~\ref{fig.3}, we consider $\hat{u}$ as a predicted (estimated) control input and $\hat{x}$ as a predicted state that is computed by the controller. In addition, we consider $\theta(k)$ and $\Phi(k)$ as two i.i.d. \emph{binary random variables} (mutually independent) with Bernoulli distributions that capture \emph{lost and received communications} between, respectively, sensor-to-controller and controller-to-actuator branches. Moreover, $\mu_\theta $ and $\mu_\Phi$ are the chance of successful packets delivery corresponding to $\theta(k)=\!1$ and $\Phi(k)=\!1$, respectively.

\noindent Since the received measurements in the controller belong to one time step in the past, the controller updates the system as follows:
\begin{equation}
\hat{x}(k+1) = \theta(k)A y(k)+ (1-\theta(k))A \hat{x}(k) +B \hat{u}(k),
\label{est}
\end{equation}
  
\noindent where $\hat{x}(k)$ is the predicted state computed by controller. By employing $y(k)$ as in~\eqref{eq.12}, $\hat{x}(k+1)$ can be rewritten as
\begin{equation}
\hat{x}(k+1) = \theta(k) A x(k)+ (1-\theta(k))A \hat{x}(k) +B \hat{u}(k) + \theta(k) A w_2(k).
\end{equation}

\begin{figure}[t]
	\includegraphics[scale=0.85]{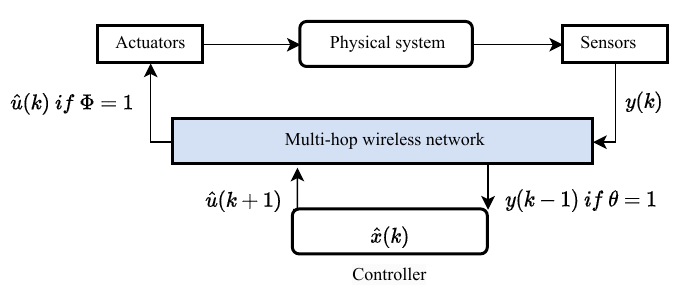}\vspace{0.1cm}
	\centering
	\caption{Communications between physical systems and controllers under a multi-hop wireless network. Symbols $\theta$ and $\Phi$ are two i.i.d. binary random variables that indicate \emph{lost and received communications} between, respectively, sensor-to-controller and controller-to-actuator branches.}
	\label{fig.3}\vspace{0.3cm}
\end{figure}
\!\!\noindent In particular, if the measurement message is successfully transmitted, controller computes the state prediction based on $y(k-1)=x(k-1)+w_2(k-1)$. If the measurement massage is lost, controller propagates the previous prediction as $\hat{x}(k-1)$. In our underlying framework, we design an output-feedback controller as
  \begin{equation}\label{eq.6}
    u(k)=Fy(k) = F(x(k)+w_2(k)),
\end{equation} 
where $F$ is a control matrix of an appropriate dimension. 

In multi-hop wireless communication networks, both measurement and control input massages may have some delay. Inspired by \cite{ref17}, we consider the actual control input $u(k)$ depending on the measurement from two time steps ago, \emph{i.e.,} $y(k-2)$, if there are no losses of messages in the network. If a measurement or control input message is lost, the actual control input $u(k)$ remains unchanged given a zero-order hold embedded in the actuator. It implies that, if a control input message is lost, the actuator holds the last input value and applies it until the next input value is received, \emph{i.e.,}
\begin{equation}
  {u}(k) = \Phi(k-1)\hat{u}(k) + (1-\Phi(k-1))u(k-1).
 \label{eq.7}
\end{equation}
Another time step is required for the predicted control input to reach to the physical system. To calculate the next control input $\hat{u}(k+1)$, we proceed with propagating the system one additional time step as
\begin{equation}
\hat{u}(k+1) = F(A\hat{x}(k)+B\hat{u}(k)).
 \label{eq.8}
\end{equation}

	\begin{remark}\label{Remark}
		If there is no massage loss in the network, including sensors-to-controller and controller-to-actuators wireless links, equation \eqref{eq.7} can be rewritten as ${u}(k) = \Phi(k-1)\hat{u}(k)$. Accordingly, using equations~\eqref{eq.8} and~\eqref{est}, one has ${u}(k) = F\big(A(Ay(k-2)+B\hat{u}(k-2))+ B\hat{u}(k-1)\big)$.
\end{remark}

In our proposed setting, we augment the real and estimated states and inputs of dt-SLS as $z= [x;\hat{x};u;\hat{u}]$, $z \in \mathbb{R}^\kappa$, with $\kappa:=2(n+m)$. Furthermore, we augment the process and measurement noises as $w=[w_{1},w_{2}]$, $w \in \mathbb{R}^{2n}$.
Then the evolution of the \emph{augmented dt-SLS} can be described as\footnote{The full derivation for obtaining the augmented system in~\eqref{augment} is provided in Appendix.}
\begin{equation}\label{augment}
    z(k+1) = \tilde{A}(k) z(k) + \tilde{E}(k)w(k),
\end{equation}  
where
\begin{align}\label{ag}
 &\tilde{A}(k) = \begin{bmatrix}
A &  0 & B &  0  \\
\theta(k) A & (1-\theta(k))A & 0 & B  \\
0 & \Phi(k) F A & (1-\Phi(k))\mathbf{I}_{m} & \Phi(k) FB\\
0 & FA & 0 & FB 
\end{bmatrix}\!\!,
&\tilde{E}(k) = \begin{bmatrix}
1 & 0 \\
0 & \theta(k) A\\
0 & 0\\
0 & 0
\end{bmatrix}\!\!.
\end{align}
We transform $\theta(k)$ by rewriting it as $\theta(k) = \mu_{\theta} (1 - \delta_\theta(k))$, where $\delta_\theta(k)$ is a new binary random variable with values of  either $1$ or $1- \frac{1}{\mu_{\theta}}$. Similarly, we transform $\Phi(k)$ as $\Phi(k) = \mu_\Phi (1 - \delta_\Phi(k))$, where $\delta_\Phi(k)$ is either $1$ or $1- \frac{1}{\mu_{\Phi}}$. These two transformations mainly facilitate our computation process by allowing us to decompose $\tilde{A}(k)$ as 
\begin{equation}
\tilde{A}(k) = \tilde{A}_{0} + \sum^{2}_{i=1} \tilde{A}_{i} \delta_{i}(k),
\label{aug1}
\end{equation}
where 
\begin{align*}
 &\tilde{A}_0=
\begin{bmatrix}
A &&  0 && B &&  0    \\
\mu_\theta A &&(1-\mu_\theta) A && 0 && B  \\
0 && \mu_\Phi F A && (1-\mu_\Phi) \mathbf{I}_{m}&& \mu_\Phi FB  \\
0 && FA && 0 && FB 
\end{bmatrix}\!\!, 
\quad\tilde{A}_1=
\begin{bmatrix}
0 &&  0 && 0 &&&&  0    \\
-\mu_\theta A && \mu_\theta A && 0 &&&& 0  \\
0 && 0 && 0 &&&& 0 \\
0 && 0 && 0 &&&& 0
\end{bmatrix}\!\!,\\
&\tilde{A}_2=
\begin{bmatrix}
0 &&  0 && 0 &&  0 &   \\
 0 && 0 && 0 && 0  \\
0 && -\mu_\Phi F A && \mu_\Phi\mathbf{I}_{m} && -\mu_\Phi FB  \\
0 && 0 && 0 && 0
\end{bmatrix}\!\!,
\quad\delta_i(k)=
\begin{cases}
\delta_1(k) = \delta_\theta(k), \\
\delta_2(k) = \delta_\Phi(k).  
\end{cases}  
\end{align*}

\noindent Then probabilities of $\delta_{\theta}(k)$ and $\delta_{\Phi}(k)$ can be defined as

	\begin{align*}
	&\PP\big[\delta_{\theta}(k)\big]\!\!:=\! \begin{cases}
	\PP\big[\delta_{\theta}(k)\!=\! 1\big] \!=\! 1\!-\!\mu_\theta, \\
	\PP\big[\delta_{\theta}(k)\!=\! 1\!-\!\frac{1}{\mu_\theta}\big] \!=\! \mu_\theta, \\
	\end{cases}
	&\PP\big[\delta_{\Phi}(k)\big]\!\!:=\! \begin{cases}
	\PP\big[\delta_{\Phi}(k)\!=\! 1\big] \!=\! 1\!-\!\mu_\Phi, \\
	\PP\big[\delta_{\Phi}(k)\!=\! 1\!-\!\frac{1}{\mu_\Phi}\big] \!=\! \mu_\phi, \\
	\end{cases} 
	\end{align*}

\noindent where $\EE\big[\delta_{\theta}(k)\big]=0, \EE\big[\delta_{\Phi}(k)\big]=0$, $\text{Var}\big[\delta_{\theta}(k)\big] =\frac{1}{\mu_\theta}-1$, and $\text{Var}\big[\delta_{\Phi}(k)\big]=\frac{1}{\mu_\Phi}-1$. Similarly, using the proposed transformations, the noise matrix $\tilde{E}$ can be decomposed as
 \begin{equation}
\tilde{E} = \tilde{E}_{0} + \sum^{2}_{i=1} \tilde{E}_{i} \delta_{i},
\label{aug2}
\end{equation}
where
\begin{align*}
\tilde E_0\!=\!
\begin{bmatrix}
1 &  0 \\
0 & \mu_\theta A\\
0 & 0\\
0 & 0
\end{bmatrix}\!\!, 
\quad\tilde E_1\!=\!
\begin{bmatrix}
0 &  0 \\
0 & -\mu_\theta A\\
0 & 0\\
0 & 0
\end{bmatrix}\!\!,
\quad\tilde E_2\!=\! \mathbf{0}_{2(n+m)\times2n}.
\end{align*}

In the next section, we define control barrier certificates for the \emph{augmented dt-SLS}. We then leverage CBC and formally quantify a lower bound on the satisfaction probability of the safety specification in the sense of Definition~\ref{Safety}. Moreover, we provide a systematic approach to design a CBC and its associated controller based on solving sum-of-squares optimization problems and matrix inequalities.
\section{Control Barrier Certificates}{\label{sec:CBC}}
Here, we first present control barrier certificates for the \emph{augmented dt-SLS} in~\eqref{augment} as the following.
\begin{definition}\label{barrier}
Consider a dt-SLS $\mathbb{S}=(X,U,Y,A,B, w_1, w_2)$ with $X_{0}, X_{1} \subseteq X$ as its initial and unsafe sets, respectively. Consider a multi-hop wireless communication network with estimated states and inputs $\hat x \in X, \hat u \in U$. Let $u(k) = F(x(k) + w_2(k))$ as proposed in~\eqref{eq.6}. Then a function $\mathbb{B}\!: X\times X \times U\times U \to \mathbb{R}_{\geq0}$ is said to be a control barrier certificate (CBC) for the augmented dt-SLS in~\eqref{augment} if there exist constants $\beta,\eta \in \mathbb{R}_{\geq0}$, with $ \beta > \eta $, and $c\in \mathbb{R}_{>0}$, such that
\begin{itemize}
\item $\forall z \in X_{0}\times X \times U\times U\!:$
\begin{align}\label{con}
\mathbb{B}(z) \leq \eta,
\end{align}
\item $\forall z \in X_{1}\times X \times U\times U\!:$
\begin{align}\label{con1}
\mathbb{B}(z) \geq \beta,
\end{align}
\item $\forall z \in X\times X \times U\times U\!:$
\begin{equation}\label{eq.21}
\EE\Big[ \mathbb{B}(z(k+1)) \,\big|\, z(k) \Big] \leq \mathbb{B}(z(k))+c.
\end{equation}
\end{itemize}
\end{definition}

In the following theorem, under the fundamental results of~\cite{ref26}, we employ Definition~\ref{barrier} and quantify an upper bound on the probability that trajectories of dt-SLS reach an unsafe region within a finite time horizon.

\begin{theorem}\label{Kushner}
Consider a dt-SLS $\mathbb{S}=(X,U,Y,A,B,w_1,w_2)$ as in Definition~\ref{dt-SLS}. Let there exist a CBC $\mathbb{B}$ for the augmented dt-SLS in~\eqref{augment} as in Definition~\ref{barrier}. Then the probability that the solution process of $\mathbb{S}$ starting form any initial state $x_0 \in X_0$ reaches $X_1$ within a time horizon $k \in [0,\mathcal T]$ is
\begin{equation}
    \PP \Big\{x_{x_0 u}(k) \in X_1~\text{for some}~ k \in [0,\mathcal T] \,\big|\, x_0,u \Big\} \leq  \frac{\eta + c \mathcal T}{\beta}.
    \label{zet}
\end{equation} 
\end{theorem}

\begin{proof}
  According to condition (\ref{con1}), $$X_{1}\times X \times U\times U \subseteq \big\{z \in X\times X \times U\times U \enspace |\enspace  \mathbb{B}(z) \geq \beta \big\}.$$ Then we have
\begin{align*}  
 &\PP \Big\{x_{x_0 u}(k) \in X_1\times X \times U\times U~\text{for some}~ k \in [0,\mathcal T] \,\big|\, x_0,u \Big\} 
 \leq \PP\Big\{\sup_{0\leq k\leq \mathcal T} \enspace \mathbb{B}(z(k)) \geq \beta \,\big|\, x_0,u \Big\}.
\end{align*}

\noindent Then the proposed bound in (\ref{zet}) can be acquired by applying \cite[Theorem 3]{ref26} to (\ref{zet}) and employing, respectively, conditions~\eqref{eq.21} and~\eqref{con} in Definition~\ref{barrier}.$\hfill$
\end{proof}

The results in Theorem~\ref{Kushner} provide an upper bound for the solution process of $\mathbb{S}$ starting form the initial region $X_0$ reaches the unsafe region $X_1$ within the time horizon $k \in [0,\mathcal T]$. Now the satisfaction of the safety specification $\Upsilon$ within the same time horizon $\mathcal T$ can be guaranteed as
\begin{align*}
\PP \Big\{\mathbb{S}\models\Upsilon\Big\}\ge 1-\frac{\eta + c \mathcal T}{\beta}.
\end{align*}

\subsection{Computation of CBC and the Safety Controller}{\label{subsec:CBCC}}
In this subsection, we consider the structure of our control barrier certificate to be quadratic in the form of $\mathbb{B}(z) = z^{\top}Pz$, where $P\in\mathbb{R}^{\kappa \times\kappa} $, $\kappa=2(n+m)$, is a positive-definite matrix. We also assume that both noises $w_1$ and $w_2$ in~\eqref{eq.1} are random variables (mutually independent) with \emph{standard normal distribution}, \emph{i.e.,} $w_1\!\sim\! \mathcal{N}\!(0,\Sigma_{w_1}\!)$, $w_2\!\sim\! \mathcal{N}\!(0,\Sigma_{w_2}\!)$. Then, we propose a systematic approach to design a CBC and its associated safety controller enforcing the safety requirements of dt-SLS $\mathbb{S}$. In our setting, the first two conditions in Definition~\ref{barrier} can be reformulated as a sum-of-square (SOS) optimization problem, while the third condition is presented as some matrix inequalities. To do so, we first raise the following assumption.

\begin{assumption}\label{Assum}
	Suppose that for some constants $\mu_\theta,\mu_\Phi \in[0,1]$, there exist matrices $P\succ0$ and $F$ of appropriate dimensions such that the matrix inequality in~\eqref{LMI} holds. 
\end{assumption}

	\begin{figure*}
	\begin{align}\notag
	&\begin{bmatrix}
	A &&  0 && B &&  0 \\
	\mu_\theta A && (1\!-\!\mu_\theta) A && 0 && B  \\
	0 &&\mu_\Phi F A && (1\!-\!\mu_\Phi)\mathbf{I}_{m} && \mu_\Phi F B  \\
	0 && FA && 0 && F B 
	\end{bmatrix}^{\top}
	\!\!\!{P}
	\begin{bmatrix}
	A &&  0 && B &&  0 \\
	\mu_\theta A && (1\!-\!\mu_\theta) A && 0 && B  \\
	0 && \mu_\Phi F A && (1\!-\!\mu_\Phi)\mathbf{I}_{m} && \mu_\Phi FB  \\
	0 && FA && 0 && F B 
	\end{bmatrix} \\\notag
	&+(1-\mu_\theta)^2\begin{bmatrix}
	0 &&  0 && 0 &&  0\\
	-A &&  A && 0 && 0  \\
	0 && 0 && 0 && 0 \\
	0 && 0 && 0 && 0 \\
	\end{bmatrix}^{\top}
	\!\!\!{P}
	\begin{bmatrix}
	0 &&  0 && 0 &&  0 \\
	-A &&  A && 0 && 0  \\
	0 && 0 && 0 && 0  \\
	0 && 0 && 0 && 0
	\end{bmatrix}\\\label{LMI}
	&+(1-\mu_\Phi)^2\begin{bmatrix}
	0 &&  0 && 0 &&  0 \\
	0 && 0 && 0 && 0  \\
	0 &&  -FA && \mathbf{I}_{m}&& -FB  \\
	0 && 0 && 0 && 0
	\end{bmatrix}^{\top}
	\!\!\!{P} \begin{bmatrix}
	0 &&  0 && 0 &&  0 \\
	0 && 0 && 0 && 0  \\
	0 && -FA && \mathbf{I}_{m} && -FB  \\
	0 && 0 && 0 && 0 
	\end{bmatrix}-P \preceq  0.
	\end{align}
\rule{\textwidth}{0.1pt}
\end{figure*}

We now present the main result of the work for designing a CBC and its associated safety controller via the following theorem.
\begin{theorem} \label{Theorem:3}
	Consider a dt-SLS $\mathbb{S}=(X,U,Y,A,B,w_1,w_2)$ in Definition~\ref{dt-SLS}. Let Assumption~\ref{Assum} hold. Assume that $X_0,X_1,X$ and $U$ are all continuous sets defined by vectors of polynomial inequalities (\emph{i.e.,} semi-algebraic sets) as $X_0 = \{ x \in \mathbb{R}^n \,\big|\, g_0(x) \geq 0 \}$, $X_1 = \{ x \in \mathbb{R}^n \,\big|\, g_1(x) \geq 0 \}$, $X = \{ x \in \mathbb{R}^n \,\big|\, g(x) \geq 0 \}$, and $U = \{u\in \mathbb{R}^m \,\big|\, g_u(x) \geq 0 \}$. Suppose there exist a quadratic function $\mathbb{B}(z) = z^{\top}Pz$, constants $\beta,\eta \in \mathbb{R}_{\geq0}$, with $ \beta > \eta $, $c\in \mathbb{R}_{>0}$, and vectors of SOS polynomials $l_0(x)$, $l_1(x), l(x),l_u(x)$, such that the following expressions are SOS polynomials:
\begin{align} \label{eq5} 
-&z(x)^{\top}Pz(x) - \big[l_0(x);l(x); l_u(x);l_u(x)\big]^\top  \big[g_0(x);g(x);g_u(x);g_u(x)\big] + \eta, \\\label{eq6} 
&z(x)^{\top}Pz(x) - \big[l_1(x);l(x); l_u(x);l_u(x)\big]^\top \big[g_1(x);g(x);g_u(x);g_u(x)\big] - \beta.
\end{align}
Then $\mathbb{B}(z) = z^{\top}Pz$ is a CBC and $F$ is a safety controller satisfying conditions~\eqref{con}-\eqref{eq.21}, with $P$ and $F$ as in~\eqref{LMI}, and
\begin{align}\label{constant}
c=~\!&\text{Tr}\Big(\begin{bmatrix}
1 &&  0 \\
0 && \mu_\theta A \\
0 && 0 \\
0 && 0 
\end{bmatrix}^{\top}
\!\!\!{P}
\begin{bmatrix}
1 &&  0 \\
0 && \mu_\theta A \\
0 && 0 \\
0 && 0 
\end{bmatrix}\Sigma_w
+
\begin{bmatrix}
0 &&  0  \\
0 && (\mu_\theta-1) A \\
0 && 0  \\
0 && 0
\end{bmatrix}^{\top}
\!\!\!{P}
\begin{bmatrix}
0 &&  0  \\
0 && (\mu_\theta-1) A \\
0 && 0\\
0 && 0
\end{bmatrix}\Sigma_w\Big),
\end{align}
where $\Sigma_w$ is the covariance matrix of the augmented noise $w$. 
\end{theorem}

\begin{proof}
We first show that the matrix inequality in~\eqref{LMI} implies the satisfaction of condition~\eqref{eq.21}. By applying function $\mathbb{B}$ on the augmented dt-SLS in~\eqref{augment}, one has
\begin{align*}     
    \mathbb{B}(&z(k+1)) = z(k)^{\top} \Tilde{A}(k)^{\top} {P} \Tilde{A}(k) z(k) + 2w(k)^{\top} \tilde{E}(k)^{\top} {P} \Tilde{A}(k) z(k)
    + w(k)^{\top} \tilde{E}(k)^{\top} {P} \tilde{E}(k)w(k).
\end{align*} 

By taking the expected value of $\mathbb{B}$ with respect to different sources of noises and since $w\!\sim\! \mathcal{N}\!(0,\Sigma_{w}\!)$, one has
\begin{align}\label{eq.48}
 &\EE\Big[\mathbb{B}(z(k+1)) \mid z(k)\Big] = \EE\left[z(k)^{\top} \Tilde{A}(k)^{\top} {P} \Tilde{A}(k) z(k)\right] + \EE\left[w(k)^{\top} \tilde{E}(k)^{\top} {P} \tilde{E}(k)w(k)\right]\!\!.
 \end{align}
By substituting $\tilde{A}(k) = \tilde{A}_{0} + \delta_{\theta}\tilde{A}_{1} + \delta_{\Phi}\tilde{A}_{2}$ in~\eqref{eq.48}, we have
\begin{align*}
\EE&\Big[\mathbb{B}(z(k+1)) \mid z(k)\Big] =
\EE\!\left[z(k)^{\top}\tilde{A}_0^{\top}P\tilde{A}_0z(k)\right]+\EE\left[z(k)^{\top}\delta_{\theta}\tilde{A}_1^{\top} P\delta_{\theta}\tilde{A}_1z(k)\right]\!+\!\EE\left[z(k)^{\top}\delta_{\Phi}\tilde{A}_2^{\top} P\delta_{\Phi}\tilde{A}_2z(k)\right]\\ &+\EE\left[w(k)^{\top} \tilde{E}(k)^{\top} P \tilde{E}(k)w(k)\right]
= z(k)^{\top} \Big( \tilde{A}_0^{\top} P\tilde{A}_0 +\EE\big[\delta_{\theta}^2\big]\tilde{A}_1^{\top} P \tilde{A}_1+\EE\big[\delta_{\Phi}^2\big]\tilde{A}_2^{\top}P\tilde{A}_2 \Big)z(k)  \\ &+\text{Tr}\big(\EE\big[\tilde{E}(k)^{\top} {P} \tilde{E}(k)\big] \underbrace{\EE\left[w(k)w(k)^{\top}\right]}_{\Sigma_w}\big).
\end{align*}
By employing $\EE\big[\delta_{\theta}^2(k)\big] =\frac{1}{\mu_\theta}-1$, $\EE\big[\delta_{\phi}^2(k)\big] =\frac{1}{\mu_\phi}-1$, $\tilde{E}(k) = \tilde{E}_{0} + \delta_{\theta}\tilde{E}_{1} + \delta_{\Phi}\tilde{E}_{2}$, and the matrix inequality in~\eqref{LMI}, one can readily show that  
\begin{equation*}
\EE\Big[ \mathbb{B}(z(k+1)) \,\big|\, z(k) \Big] \leq \mathbb{B}(z(k))+c,
\end{equation*} 
with $c$ as in~\eqref{constant}.

\noindent We now proceed with showing that conditions~\eqref{eq5},~\eqref{eq6} imply the satisfaction of conditions~\eqref{con},~\eqref{con1}, as well. Since~\eqref{eq5} is an SOS polynomial, we have 
\begin{align*}
    0\leq -z(x)^{\top}Pz(x) - \big[l_0(x);l(x); l_u(x);l_u(x)\big]^\top \big[g_0(x);g(x);g_u(x);g_u(x)\big] + \eta.
\end{align*}
 Given that the term $\big[l_0(x);l(x);l_u(x);l_u(x)\big]^\top\big[g_0(x);g(x);g_u(x);g_u(x)\big]$ is non-negative over $X_{0}\times X \times U\times U$, condition~\eqref{eq5} implies satisfaction of condition~\eqref{con}. Similarly, we can show that condition \eqref{eq6} implies \eqref{con1}. Hence, $\mathbb{B}(z) = z^{\top}Pz$ is a CBC and $F$ is a safety controller satisfying conditions~\eqref{con}-\eqref{eq.21}, which concludes the proof.
\end{proof}

We now propose the following algorithm to elucidate the necessary procedure for computing CBC and its associated safety controller.

\begin{algorithm}[H]
	\caption{CBC and safety controller computation}
	\begin{algorithmic}[1]
		\Require dt-SLS $\mathbb{S}=(X,U,Y,A,B,w_1,w_2)$, initial set $X_0$, unsafe set $X_1$, time horizon $\mathcal T$, chances of successful packet transmissions $\mu_\theta$ and $\mu_\phi$
		\State Solve matrix inequality in~\eqref{LMI} using semi-definite programming (SDP) solver \textsf{SeDuMi} to compute  $c$, $P\succ0$ and feedback controller gain $F$
		\State{Solve SOS programming via software tool \textsf{SOSTOOLS}, using $\mathbb{B}(z)$, to design level sets $\eta, \beta$ by satisfying conditions~\eqref{eq5},~\eqref{eq6}}
		\State Compute $\PP \Big\{\mathbb{S}\models\Upsilon\Big\}\ge 1-\frac{\eta + c \mathcal T}{\beta}$
		\Ensure Barrier certificate $\mathbb{B}(z)$, controller $F$, safety guarantee $\PP \Big\{\mathbb{S}\models\Upsilon\Big\}\ge 1-\frac{\eta + c \mathcal T}{\beta}$
	\end{algorithmic}
	\label{Alg}
\end{algorithm}

\section{Case Study: Electrical Motor}{\label{sec:case}}
We demonstrate the efficacy of our proposed results by applying them to an electrical motor. Under a multi-hop wireless communication network, chances of successful packet transmissions between sensor-to-controller and controller-to-actuator branches are given as $\mu_\theta=0.9$ and $\mu_\phi=0.9$, respectively (Fig.~\ref{fig.4}). The main goal is to design a safety controller to maintain the electrical current of the motor in a safe set within a finite time horizon while respecting the multi-hop wireless network requirements. The evolution of the electrical motor can be described by the following dynamics~\cite{ref20}: 

\begin{equation*}
\mathbb{S}\!:
    \begin{cases}
        x(k+1) = Ax(k) + Bu(k) + {w_1}(k), \\
      y(k) = x(k) + w_2(k),
    \end{cases}  
    \label{eq.2}
\end{equation*}

\noindent with
\begin{align*}
&A \!=\! \begin{bmatrix} \frac{L_d}{R}\frac{1}{e^{\frac{R}{L_d}T_s}} & L_q\frac{\omega_{el}}{R}(1\!-\!e^{-\frac{R}{L_d}T_s})\\
-L_d\frac{\omega_{el}}{R}(1\!-\!e^{-\frac{R}{L_q}T_s}) & \frac{L_q}{R}\frac{1}{e^{\frac{R}{L_q}T_s}} \end{bmatrix}\!\!, \\&B \!=\! \begin{bmatrix} \frac{1}{R}(1\!-\!e^{-\frac{R}{L_d}T_s}) & \frac{L_q}{R}(1\!-\!e^{-\frac{R}{L_d}T_s}) \\
\frac{L_d}{R}(1\!-\!e^{-\frac{R}{L_q}T_s}) & \frac{1}{R}(1\!-\!e^{-\frac{R}{L_q}T_s}) \end{bmatrix}\!\!,
\end{align*}
\!\!where $T_s$ is the sampling time, $R = 0.025~\text{[Ohm]}$ is the resistance of the motor's stator winding, $\omega_{el} = 6283.2~\text{[rad/s]}$ is electrical angular frequency of the AC power supply applied to the motor, and $L_d = 0.0001$ [ht]\footnote{``ht'' stands for Henry turn as inductance unit.} and $L_q = 0.00012~\text{[ht]}$ are the d-axis and q-axis inductances of a three-phase AC motor. In addition, ${w_1}=[w_{11};w_{12}]$, ${w_2}=[w_{21};w_{22}]$, and $x = [x_d;x_q]$, where $x_d$ and $x_q$ are currents in the $d$ and $q$ coordinates over time, respectively. The control input is $u = [u_d;u_q]$, where $u_d$ and $u_q$ represent voltages applied to the motor in the $d$ and $q$ coordinates, respectively.

\begin{figure}
	\includegraphics[scale=0.85]{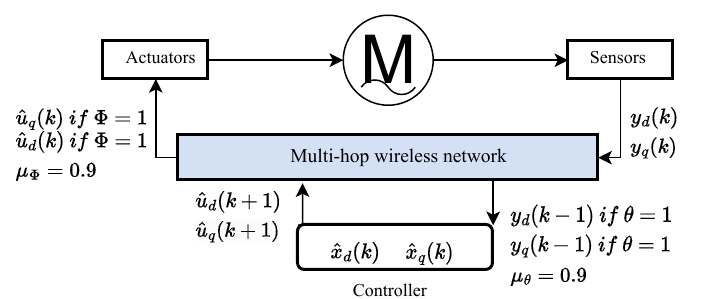}\vspace{0.02cm}
	\centering
	\caption{Communications between an electrical motor \textbf{\textsf{M}} and controller, under a multi-hub wireless communication network, with chances of successful packet transmissions between sensor-to-controller and controller-to-actuator branches as $\mu_\theta=0.9$ and $\mu_\phi=0.9$, respectively.}
	\label{fig.4}\vspace{0.3cm}
\end{figure}

The regions of interest are given by $X=[-2, 2]^2$, $X_0=[-0.2, 0.2]^2$, and $X_1=[-2,-1.2]\times[-2,2]\cup[1.2,2]\times[-2,2]$. The ultimate objective is to design a CBC for the electrical motor and its associated safety controller such that the system's states stay within the safe set $[-1.2, 1.2]\times[-2,2]$.

To achieve this, we first employ semi-definite programming (SDP) solver \textsf{SeDuMi} and design barrier matrix $P$ together with its corresponding controller matrix $F$ in~\eqref{LMI} as in~\eqref{p}.
We now employ barrier matrix $P$ and design its corresponding initial and unsafe level sets $\eta, \beta$ by satisfying conditions~\eqref{eq5},~\eqref{eq6} via software tool \textsf{SOSTOOLS}, as $\eta=0.0001306, \beta=0.7233$. In addition, the constant $c$ is computed as $0.000166$.

\begin{figure*}
	\begin{align}\label{p}
	&P\!=\! \begin{bmatrix} 
	0.92 & 0.68 & -0.58 & -0.031 & 0.34 & -0.061 & -0.073 & -0.074\\ 0.68 & 3.7 & -0.41 & -0.85 & 0.21 & 0.33 & 0.036 & 0.12\\ -0.58 & -0.41 & 0.94 & 0.22 & -0.088 & -0.056 & 0.43 & -0.034\\ -0.031 & -0.85 & 0.22 & 1.2 & -0.071 & 0.066 & 0.23 & 0.14\\ 0.34 & 0.21 & -0.088 & -0.071 & 1.5 & -0.013 & -0.85 & -0.093\\ -0.061 & 0.33 & -0.056 & 0.066 & -0.013 & 0.91 & -0.13 & -0.13\\ -0.073 & 0.036 & 0.43 & 0.23 & -0.85 & -0.13 & 1.5 & -0.053\\ -0.074 & 0.12 & -0.034 & 0.14 & -0.093 & -0.13 & -0.053 & 0.88
	\end{bmatrix}\!\!,~F\!=\! \begin{bmatrix}
	-0.68  &  -0.70\\
	0.79  & -0.60
	\end{bmatrix}\!\!.
	\end{align}
\rule{\textwidth}{0.1pt}
\end{figure*}

Now by employing Theorem~\ref{Kushner}, we guarantee that all state trajectories of $\mathbb{S}$ starting from $X_0 = [-0.2, 0.2]^2$ remain in the safe set $[-1.2, 1.2]\times[-2,2]$ during time horizon $\mathcal{T}=100$ with a guaranteed probability of at least $97.68\%$. Closed-loop state trajectories of the electrical motor under the designed controller $F$ in~\eqref{LMI} are depicted in Figs.~\ref{fig:b},~\ref{fig:a},
with $10$ different noise realizations. As it can be observed, all trajectories of the electrical motor respect the desired safety specification.	
\begin{figure}
  \centering
  	\includegraphics[scale=0.61]{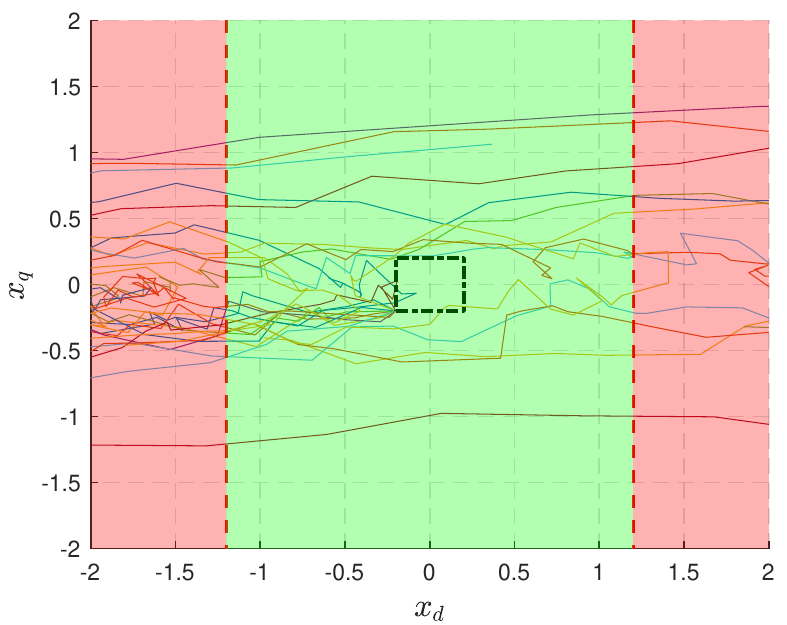}\hspace{0.3cm}
    \includegraphics[scale=0.61]{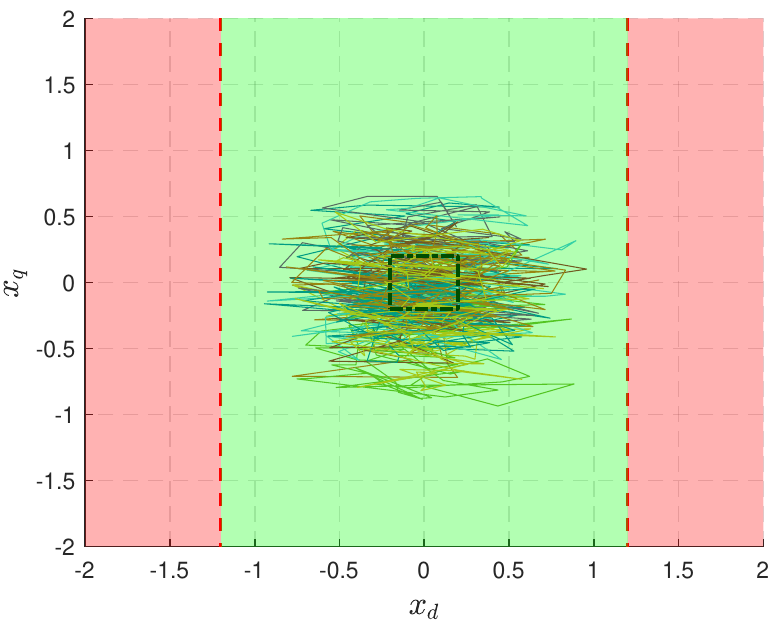}
  \caption{Closed-loop state trajectories of the electrical motor with $10$ different noise realizations starting from the initial set $[-0.2, 0.2]^2$ without controller (left figure) and with controller (right figure). Green and red boxes are safe and unsafe regions, respectively. The central rectangle indicates the initial set.}
  \label{fig:b}
\end{figure}
\begin{figure}
\begin{center}
\includegraphics[scale=0.57]{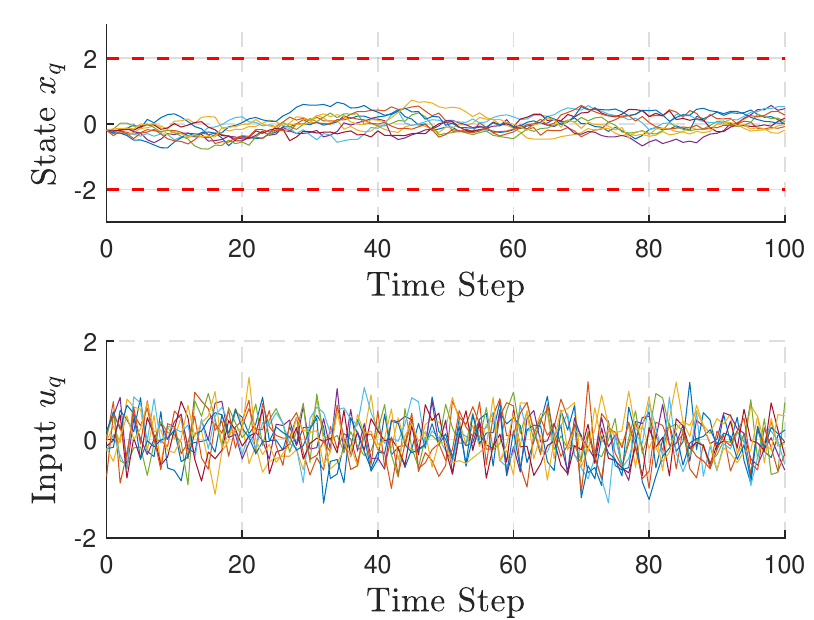}
\includegraphics[scale=0.57]{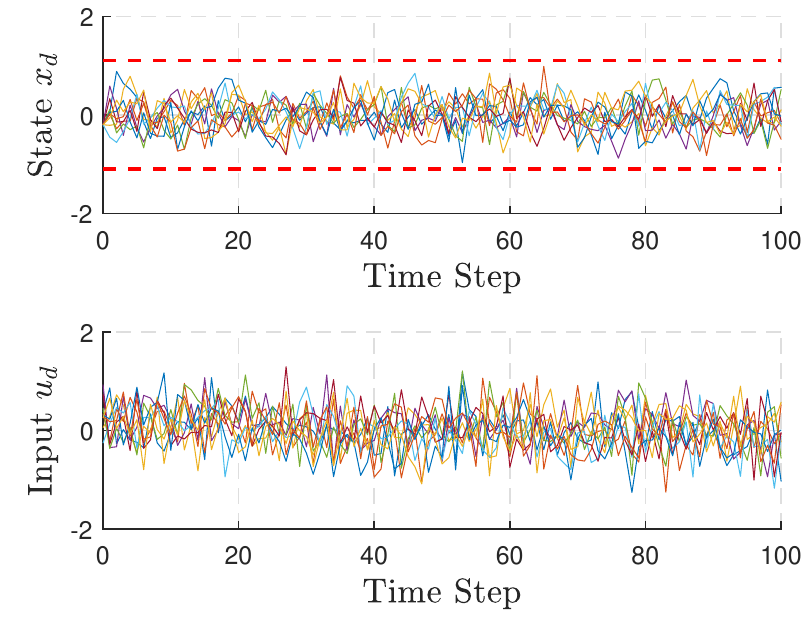}
\caption{Closed-loop state trajectories of the electrical motor under the synthesized controller with $10$ different noise realizations over $100$ time steps. Red dashed lines are boundaries of the safe set.}
\label{fig:a}
\end{center}
\end{figure}
\section{Conclusion}\label{conclude}
In this work, we proposed a formal approach for safety controller design of stochastic control systems while considering wireless communication networks between sensors, controllers, and actuators. Under the notion of control barrier certificates, we provided a lower bound on the satisfaction probability of the safety property over a finite time horizon for wirelessly connected stochastic control systems. We proposed a systematic approach in the form of sum-of-squares optimization and matrix inequalities for the construction of CBC and its associated controller. We illustrated the efficacy of our approach on a permanent magnet synchronous motor. Developing a formal approach for safety controller synthesis of \emph{networks of} wirelessly-connected stochastic control systems with \emph{nonlinear dynamics} is a potential future direction.

\bibliographystyle{alpha}
\bibliography{biblio}

\newcommand{\etalchar}[1]{$^{#1}$}
\begin{thebibliography}{MHMJZ20}

\bibitem[ACE{\etalchar{+}}19]{ref10}
A.~D. Ames, S.~Coogan, M.~Egerstedt, G.~Notomista, K.~Sreenath, and P.~Tabuada.
\newblock {C}ontrol {B}arrier {F}unctions: {T}heory and {A}pplications.
\newblock In {\em Proceedings of the 18th European Control Conference (ECC)},
  pages 3420--3431, 2019.

\bibitem[ALZ22]{anand2022small}
M.~Anand, A.~Lavaei, and M.~Zamani.
\newblock From small-gain theory to compositional construction of barrier
  certificates for large-scale stochastic systems.
\newblock {\em IEEE Transactions on Automatic Control}, 67(10):5638--5645,
  2022.

\bibitem[APLS08]{ref1}
A.~Abate, M.~Prandini, J.~Lygeros, and S.~Sastry.
\newblock Probabilistic reachability and safety for controlled discrete time
  stochastic hybrid systems.
\newblock {\em Automatica}, 44(11):2724--2734, 2008.

\bibitem[AXGT17]{ref12}
A.~D. Ames, X.~Xu, J.~W. Grizzle, and P.~Tabuada.
\newblock Control {B}arrier {F}unction {B}ased {Q}uadratic {P}rograms for
  {S}afety {C}ritical {S}ystems.
\newblock {\em IEEE Transactions on Automatic Control}, 62(8):3861--3876, 2017.

\bibitem[BK08]{baier2008principles}
C.~Baier and J.-P. Katoen.
\newblock {\em Principles of model checking}.
\newblock MIT press, 2008.

\bibitem[Cla21]{ref7}
A.~Clark.
\newblock {V}erification and {S}ynthesis of {C}ontrol {B}arrier {F}unctions.
\newblock In {\em Proceedings of the 60th IEEE Conference on Decision and
  Control (CDC)}, pages 6105--6112, 2021.

\bibitem[FZTS11]{ref24}
F.~Ferrari, M.~Zimmerling, L.~Thiele, and O.~Saukh.
\newblock Efficient network flooding and time synchronization with glossy.
\newblock In {\em Proceedings of the 10th ACM/IEEE International Conference on
  Information Processing in Sensor Networks}, pages 73--84, 2011.

\bibitem[Hes04]{ref22}
J.~P. Hespanha.
\newblock {S}tochastic {H}ybrid {S}ystems: {A}pplication to {C}ommunication
  {N}etworks.
\newblock In {\em Hybrid Systems: Computation and Control}, pages 387--401,
  2004.

\bibitem[HHHK13]{hahn2013compositional}
E.~M. Hahn, A.~Hartmanns, H.~Hermanns, and J.-P. Katoen.
\newblock A compositional modelling and analysis framework for stochastic
  hybrid systems.
\newblock {\em Formal Methods in System Design}, 43(2):191--232, 2013.

\bibitem[HK19]{7}
L.~Huang and E.-Y. Kang.
\newblock Formal {V}erification of {S}afety {\&} {S}ecurity {R}elated {T}iming
  {C}onstraints for a {C}ooperative {A}utomotive {S}ystem.
\newblock In {\em Proceedings of Fundamental Approaches to Software
  Engineering}, pages 210--227, 2019.

\bibitem[JP09]{julius2009approximations}
A.~A. Julius and G.~J. Pappas.
\newblock Approximations of stochastic hybrid systems.
\newblock {\em IEEE Transactions on Automatic Control}, 54(6):1193--1203, 2009.

\bibitem[JSZ21]{ref11}
P.~Jagtap, S.~Soudjani, and M.~Zamani.
\newblock {F}ormal {S}ynthesis of {S}tochastic {S}ystems via {C}ontrol
  {B}arrier {C}ertificates.
\newblock {\em IEEE Transactions on Automatic Control}, 66(7):3097--3110, 2021.

\bibitem[JZZ{\etalchar{+}}18]{ref25}
R.~Jacob, L.~Zhang, M.~Zimmerling, J.~Beutel, S.~Chakraborty, and L.~Thiele.
\newblock {T}{T}{W}: a time-triggered wireless design for {C}{P}{S}.
\newblock In {\em Proceedings of the Design, Automation \& Test in Europe
  Conference \& Exhibition (DATE)}, pages 865--868, 2018.

\bibitem[Kus67]{ref26}
H.~J. Kushner.
\newblock Stochastic stability and control.
\newblock Technical report, Brown Univ Providence RI, 1967.

\bibitem[LF22a]{lavaei2022compositional}
A.~Lavaei and E.~Frazzoli.
\newblock Compositional controller synthesis for interconnected stochastic
  systems with {M}arkovian switching.
\newblock In {\em 2022 American Control Conference (ACC)}, pages 4838--4843,
  2022.

\bibitem[LF22b]{lavaei2022scalable}
A.~Lavaei and E.~Frazzoli.
\newblock Scalable synthesis of finite {MDPs} for large-scale stochastic
  switching systems.
\newblock In {\em 2022 IEEE 61st Conference on Decision and Control (CDC)},
  pages 7510--7515, 2022.

\bibitem[LSAZ22]{ref4}
A.~Lavaei, S.~Soudjani, A.~Abate, and M.~Zamani.
\newblock Automated verification and synthesis of stochastic hybrid systems: A
  survey.
\newblock {\em Automatica}, 146, 2022.

\bibitem[LSF22]{ref6}
A.~Lavaei, S.~Soudjani, and E.~Frazzoli.
\newblock {S}afety {B}arrier {C}ertificates for {S}tochastic {H}ybrid
  {S}ystems.
\newblock In {\em Proceedings of the American Control Conference (ACC)}, pages
  880--885, 2022.

\bibitem[LSZ20a]{ref3}
A.~Lavaei, S.~Soudjani, and M.~Zamani.
\newblock Compositional abstraction-based synthesis for networks of stochastic
  switched systems.
\newblock {\em Automatica}, 114, 2020.

\bibitem[LSZ20b]{lavaei2020compositional}
A.~Lavaei, S.~Soudjani, and M.~Zamani.
\newblock Compositional (in) finite abstractions for large-scale interconnected
  stochastic systems.
\newblock {\em IEEE Transactions on Automatic Control}, 65(12):5280--5295,
  2020.

\bibitem[LZ22]{lavaei2022dissipativity}
A.~Lavaei and M.~Zamani.
\newblock From dissipativity theory to compositional synthesis of large-scale
  stochastic switched systems.
\newblock {\em IEEE Transactions on Automatic Control}, 67(9):4422--4437, 2022.

\bibitem[MBH{\etalchar{+}}22]{ref19}
F.~Mager, D.~Baumann, C.~Herrmann, S.~Trimpe, and M.~Zimmerling.
\newblock {S}caling beyond {B}andwidth {L}imitations: {W}ireless {C}ontrol with
  {S}tability {G}uarantees under {O}verload.
\newblock {\em ACM Trans. Cyber-Phys. Syst.}, 6(3), sep 2022.

\bibitem[MBJ{\etalchar{+}}19]{ref17}
F.~Mager, D.~Baumann, R.~Jacob, L.~Thiele, S.~Trimpe, and M.~Zimmerling.
\newblock Feedback control goes wireless: guaranteed stability over low-power
  multi-hop networks.
\newblock {\em Proceedings of the 10th ACM/IEEE International Conference on
  Cyber-Physical Systems (ICCPS '19)}, pages 97--108, 2019.

\bibitem[MHMJZ20]{ref20}
M.~Maggio, A.~Hamann, E.~Mayer-John, and D.~Ziegenbein.
\newblock Control-system stability under consecutive deadline misses
  constraints.
\newblock In {\em Proceedings of the 32nd euromicro conference on real-time
  systems (ECRTS 2020)}, pages 1--23, 2020.

\bibitem[NSZ21]{nejati2021compositional}
A.~Nejati, S.~Soudjani, and M.~Zamani.
\newblock Compositional abstraction-based synthesis for continuous-time
  stochastic hybrid systems.
\newblock {\em European Journal of Control}, 57:82--94, 2021.

\bibitem[NSZ22]{nejati2022compositional}
A.~Nejati, S.~Soudjani, and M.~Zamani.
\newblock Compositional construction of control barrier functions for
  continuous-time stochastic hybrid systems.
\newblock {\em Automatica}, 145:110513, 2022.

\bibitem[NZ22]{nejati2022dissipativity}
A.~Nejati and M.~Zamani.
\newblock From dissipativity theory to compositional construction of control
  barrier certificates.
\newblock {\em Leibniz Transactions on Embedded Systems}, 8(2):06--1, 2022.

\bibitem[PJ04]{prajna2004safety}
S.~Prajna and A.~Jadbabaie.
\newblock Safety verification of hybrid systems using barrier certificates.
\newblock In {\em Proceedings of the International Workshop on Hybrid Systems:
  Computation and Control (HSCC)}, pages 477--492, 2004.

\bibitem[PJP04]{ref5}
S.~Prajna, A.~Jadbabaie, and G.J. Pappas.
\newblock Stochastic safety verification using barrier certificates.
\newblock In {\em Proceedings of the 43rd IEEE Conference on Decision and
  Control (CDC)}, pages 929--934, 2004.

\bibitem[SAM17]{SAM17}
S.~Soudjani, A.~Abate, and R.~Majumdar.
\newblock Dynamic {B}ayesian {N}etworks for {F}ormal {V}erification of
  {S}tructured {S}tochastic {P}rocesses.
\newblock {\em Acta Informatica}, 54(2):217--242, 2017.

\bibitem[SDC21]{ref8}
C.~Santoyo, M.~Dutreix, and S.~Coogan.
\newblock A barrier function approach to finite-time stochastic system
  verification and control.
\newblock {\em Automatica}, 125, 2021.

\bibitem[WAE17]{ref13}
L.~Wang, A.~D. Ames, and M.~Egerstedt.
\newblock Safety {B}arrier {C}ertificates for {C}ollisions-{F}ree {M}ultirobot
  {S}ystems.
\newblock {\em IEEE Transactions on Robotics}, 33(3):661--674, 2017.

\bibitem[WLL18]{ref16}
B.~Wu, M.~D. Lemmon, and H.~Lin.
\newblock {F}ormal {M}ethods for {S}tability {A}nalysis of {N}etworked
  {C}ontrol {S}ystems {W}ith {IEEE} 802.15.4 {P}rotocol.
\newblock {\em IEEE Transactions on Control Systems Technology},
  26(5):1635--1645, 2018.

\bibitem[ZAG15]{zamani2015symbolic}
M.~Zamani, A.~Abate, and A.~Girard.
\newblock Symbolic models for stochastic switched systems: A discretization and
  a discretization-free approach.
\newblock {\em Automatica}, 55:183--196, 2015.

\bibitem[ZFMT13]{ref23}
M.~Zimmerling, F.~Ferrari, L.~Mottola, and L.~Thiele.
\newblock On modeling low-power wireless protocols based on synchronous packet
  transmissions.
\newblock In {\em Proceedings of the 21st International Symposium on Modelling,
  Analysis and Simulation of Computer and Telecommunication Systems}, pages
  546--555, 2013.

\bibitem[ZSR{\etalchar{+}}10]{ref9}
L.~Zhang, Z.~She, S.~Ratschan, H.~Hermanns, and E.~M. Hahn.
\newblock {S}afety {V}erification for {P}robabilistic {H}ybrid {S}ystems.
\newblock In {\em Proceedings of the Computer Aided Verification}, pages
  196--211, 2010.

\end{thebibliography}

\section{Appendix}
Here, we show that how the matrices of augmented dt-SLS in~(\ref{ag}) are obtained. By augmenting the real and estimated states and inputs of dt-SLS as $z= [x;\hat{x};u;\hat{u}]$, one has
\begin{equation*}
	 z(k+1):=\begin{cases}
	x(k+1) = A x(k) + B u(k) + w_1(k),\\
     \hat{x}(k+1) = \theta(k) A x(k)+ (1-\theta(k))A \hat{x}(k) +B \hat{u}(k) + \theta(k) A w_2(k), \\
	{u}(k+1) = \Phi(k)F(A\hat{x}(k)+ B\hat{u}(k)) + (1-\Phi(k))u(k),  \\
     \hat{u}(k+1)= F(A\hat{x}(k)+ B\hat{u}(k)).
	\end{cases}  
\end{equation*}
 Then, one can rewrite $z(k+1)$ as
\begin{align*}
\begin{bmatrix}
x(k+1) \\
\hat{x}(k+1) \\
u(k+1) \\
\hat{u}(k+1)
\end{bmatrix}
= 
\begin{bmatrix}
Ax(k) + B u(k) \\
\theta(k) A x(k)+ (1-\theta(k))A \hat{x}(k) +B \hat{u}(k)\\
 \Phi(k) F(A\hat{x}(k)+ B\hat{u}(k)) + (1-\Phi(k))u(k) \\
F(A\hat{x}(k)+B\hat{u}(k))
\end{bmatrix}
+
&\begin{bmatrix}
w_1(k) \\
 \theta(k) A w_2(k)\\
0\\
0 
\end{bmatrix}\!\!.
\end{align*}
Eventually, $z(k+1)$ can be written based on $z(k)$ and $w(k)$ as
\begin{align*}
\underbrace{\begin{bmatrix}
x(k+1) \\
\hat{x}(k+1) \\
u(k+1) \\
\hat{u}(k+1)
\end{bmatrix}}_{z(k+1)}
= 
&\underbrace{\begin{bmatrix}
A &&  0 && B &&  0 &&   \\
\theta(k) A && (1-\theta(k))A && 0 && B  \\
0 && \Phi(k) F A && (1-\Phi(k))\mathbf{I}_{m} && \Phi(k) FB  \\
0 && FA && 0 && FB 
\end{bmatrix}}_{\tilde{A}(k)}
\underbrace{\begin{bmatrix}
x(k) \\
\hat{x}(k)\\
u(k)\\
\hat{u}(k) \\
\end{bmatrix}}_{z(k)}
+\underbrace{\begin{bmatrix}
1 &  0 \\
0 & \theta(k) A\\
0 & 0\\
0 & 0
\end{bmatrix}}_{\tilde{E}(k)}
\underbrace{\begin{bmatrix}
    w_1(k) \\
    w_2(k)
\end{bmatrix}}_{w(k)}\!\!.
\end{align*}
\end{document}